\documentclass[conference, a4paper]{IEEEtran}
\IEEEoverridecommandlockouts
\usepackage{multicol}
\newcommand{\mc}{\mathcal}
\newcommand{\bm}{\boldsymbol}

\newcommand{\bth}{\boldsymbol{\theta}}

\usepackage{amsfonts,upgreek,pifont,amssymb,enumerate,color,algorithm,theorem}

\usepackage{algpseudocode}
\usepackage{cite,graphicx,subfigure,epstopdf}

\usepackage{mathtools}

\newcommand{\qedsymbol}{$\blacksquare$}
\newenvironment{proof}
    {
      \emph{Proof.}
    }
    {
      \hfill\qedsymbol
    }    

\newtheorem{Proposition}{Proposition}

\def\BibTeX{{\rm B\kern-.05em{\sc i\kern-.025em b}\kern-.08em
    T\kern-.1667em\lower.7ex\hbox{E}\kern-.125emX}}

\begin{document}
\title{Deep Reinforcement Learning for IoT Networks: Age of Information and Energy Cost Tradeoff}

% \author{Xiongwei Wu, Yiding Yu, Jun Li, Xiuhua Li, H. Vincent Poor, and P. C. Ching, 

\author{ 
Xiongwei~Wu$^{1,2}$, Xiuhua Li$^3$, Jun Li$^{4}$, P. C. Ching$^1$, H. Vincent Poor$^2$\\
$^1$Dept. of Electronic Engineering, The Chinese University of Hong Kong, Shatin, Hong Kong SAR, China\\
$^2$Dept. of Electrical Engineering, Princeton University, Princeton, USA\\
$^3$School of Big Data \& Software Engineering, Chongqing University, Chongqing, China\\
$^4$School of Electronic \&  Optical Engineering, Nanjing University of Science and Technology, Nanjing, China\\
E-mail: \{xwwu, pcching\}@ee.cuhk.edu.hk; lixiuhua1988@gmail.com; jun.li@njust.edu.cn; poor@princeton.edu
% \vspace*{-0.3cm}
%
% \thanks{This work was supported in part by the Global Scholarship Programme for Research Excellence from CUHK, and in part by the U.S. National Science Foundation under Grants CCF-0939370 and CCF-1513915.}
\thanks{This work was supported in part by the Global Scholarship Programme for Research Excellence from CUHK, and in part by the U.S. National Science Foundation under Grant CCF-1908308.}
}
\maketitle
\begin{abstract}
% Massive number of devices in Internet of Things (IoT) networks are expected to generate unprecedented traffic loads on wireless networks. 
% Rapid increase of devices in Internet of Things (IoT) networks are expected to generate heavy traffic load and make wireless networks congested. 
% A promising approach for tackling this challenge is to cache IoT sensing data close to data consumers, i.e., at edge nodes. Nevertheless, these sensing data are usually transient, and thus a temporal cache update is needed.
% In this paper, 
%  which can be quantified by the metric of age of information (AoI). 
% Edge caching has been regarded as a promising approach for alleviating traffic load in Internet of Things (IoT) networks. 
% This paper studies an Internet of Things (IoT) network, in which an edge node serves as a relay that can cache sensing data generated by IoT sensors as well as provide communication services for data consumers. A critic issue is that IoT sensing data is usually transient which necessitates temporal updates of caching content items; however, frequent cache updates could lead to considerable energy cost and challenge the lifetime of IoT sensors. 
In most Internet of Things (IoT) networks, edge nodes are commonly used as to relays to cache sensing data generated by IoT sensors as well as provide communication services for data consumers. However, a critical issue of IoT sensing is that data are usually transient, which necessitates temporal updates of caching content items while frequent cache updates could lead to considerable energy cost and challenge the lifetime of IoT sensors. To address this issue, we adopt the Age of Information (AoI) to quantify data freshness and propose an online cache update scheme to obtain an effective tradeoff between the average AoI and energy cost. Specifically, we first develop a characterization of transmission energy consumption at IoT sensors by incorporating a successful transmission condition. Then, we model cache updating as a Markov decision process to minimize average weighted cost with judicious definitions of state, action, and reward. Since user preference towards content items is usually unknown and often temporally evolving, we therefore develop a deep reinforcement learning (DRL) algorithm to enable intelligent cache updates. Through trial-and-error explorations, an effective caching policy can be learned without requiring exact knowledge of content popularity. Simulation results demonstrate the superiority of the proposed framework.  
\end{abstract}

\section{Introduction}
In the foreseeable future, billions of electronic devices (e.g., smartphones, healthcare sensors, vehicles, smart cameras, smart home appliances, etc.), are anticipated to connect to the Internet, which constitutes the {\it Internet of Things} (IoT) \cite{madakam2015internet,chen2019artificial}. Massive numbers of IoT sensors will inevitably generate a tremendous amount of data, which could impose a heavy traffic burden and make wireless networks extremely congested. To cope with these challenges, pushing network resources from the cloud to the edge of wireless networks has been deemed as a promising approach for reducing traffic burden and enhancing user quality of service (QoS) in future IoT networks. 
% By introducing computing and caching capabilities at edge nodes (e.g., Fog-RAN \cite{peng2016fog}), sensing data are able to be cached close to IoT sensors and data consumers. 
% This practice significantly reduces demand on fronthaul and backhaul and decreases latency \cite{peng2016fog,dastjerdi2016fog}. 

Considerable research attention has been devoted to caching policies to optimize communication performance metrics, e.g., transmission delay, traffic load, and power consumption \cite{wu2020joint,wu2019jointMDS,li2018hierarchical}. 
However, these studies generally focus on caching multimedia content items that are usually in-transient once they are produced. In contrast, IoT sensing data cached in edge nodes are {\it transient} and could become outdated as time goes by. Hence, previous content caching strategies could not be readily utilized to provide fresh data in IoT services. 

To capture data freshness of transient content items, the {\it Age of Information} (AoI) has  emerged as an effective performance metric. AoI is defined by the time elapsed after the generation of IoT sensing data \cite{kaul2012real,sun2017update}. It is desirable to minimize the average AoI of transient content items in the design of caching policies. On the other hand, frequent cache updates will cause considerable energy consumption at IoT sensors that usually have a rather limited battery capacity. To accommodate vast numbers of smart devices and extend the lifetime of IoT sensors, efficient cache update designs are needed to be developed towards fifth-generation (5G) and beyond communications.  

Recently, machine learning, e.g., deep reinforcement learning (DRL), has been considered as a promising tool for resource management in IoT networks \cite{chen2019artificial}. Some studies have utilized DRL to develop caching policies for IoT sensing. For instance, the study in \cite{wang2020federated} investigated a caching strategy design for the IoT by using federated DRL. However, it did not take into consideration data freshness. The authors in \cite{ma2020deep,ceran2018reinforcement} studied caching transient content items by minimizing the average AoI plus cache update cost. The cost of updating content in these studies was considered as the number of transmissions between sensors and edge nodes. This simple measurement was also used in \cite{hatami2020age} by assuming unit energy consumption per transmissions. A similar idea was also considered in \cite{ceran2019reinforcement}. Generally, the above-mentioned studies ignore the impact of time-varying content popularity. 
% {\blue Another research trend is to study data freshness in unmanned aerial vehicles (UAVs) aided sensing networks, e.g., \cite{wu2020uav,yi2020deep,abd2019deep}, through optimizing UAV trajectory.}

% {\bf Contribution:}  This paper investigates average AoI and energy cost tradeoff in IoT networks, where an edge node maintains a cache unit and serves as a relay to provide service for data consumers. We consider the scenario that IoT sensors generate inhomogeneous size of transient content items that need to upload to the edge node via wireless links; and user preferences towards these content items exhibit unknown and temporal dynamics. Thus, we utilize DRL to learning caching update policy that adapts to dynamic features over the network. The main contributions of this work are summarized as follows.
% {\bf Contribution:} 
% \subsection{Contributions}
% We consider the scenario that IoT sensors generate inhomogeneous size of transient content items that need to upload to the edge node via wireless links; and user preferences towards these content items exhibit unknown and temporal dynamics. Thus, we utilize DRL to learning caching update policy that adapts to dynamic features over the network. The main contributions of this work are summarized as follows.
 This paper considers the scenario in which IoT sensors may produce the inhomogeneous size of transient content items that are required to be stored at the edge node via wireless links; meanwhile, the statistics of user requests towards IoT sensing data are uncertain and  temporally evolving. We propose an intelligent policy to attain an equitable tradeoff between the average AoI and cache update cost, which is quantified as transmission energy consumption rather than the number of cache updates. 
The main contributions of this work are summarized as follows.
\begin{itemize}
  % \item We model online cache update problem in IoT network as a Markov dynamic process (MDP) with the goal of minimizing average weighted cost incorporating average AoI for satisfying user requests and corresponding transmission energy consumption at sensors.
  % \item We study the issue of how to preserve data freshness while reducing energy consumption at IoT sensors corresponding to cache update. Particularly, we formulate a Markov dynamic process (MDP) with the goal of minimizing long-term average weighted cost.
  % \item We formulate an online cache update problem in an IoT network as a Markov decision process (MDP) with the goal of minimizing average weighted cost incorporating average AoI for satisfying user requests and corresponding transmission energy consumption at sensors.
  \item We investigate the issue of how to preserve data freshness while reducing energy consumption at IoT sensors. Particularly, we consider cache updating given the inhomogeneous size of IoT sensing data, characteristics of wireless channels, and time-varying content popularity. To tackle the resulting decision-making in such a complex and dynamic environment, we propose a novel DRL-based framework. 
  \item We develop a characterization of transmission energy consumption for uploading sensing data from IoT sensors to the edge node via wireless links. Particularly, we consider a realistic condition that wireless transmissions are effective when the received signal-to-noise ratio (SNR) is beyond a certain threshold. 
  \item We conduct simulations to demonstrate that the proposed deep Q-network (DQN) algorithm can significantly reduce energy cost while slightly compromising average AoI. Under the scenario being studied, transmission energy consumption at IoT sensors witnesses a reduction of $52.7\%$, whereas average AoI is increased by 2.41 epochs compared with AoI-oriented results.
\end{itemize}
The remainder of this paper is organized as follows. Section II introduces the system model. Section III presents the MDP problem formulation, and Section IV develops a DRL-based algorithm. Section V presents simulation results, and Section VI concludes the paper.

\section{System Model}
\subsection{System Operation}
As illustrated in Fig. \ref{system_iot}, we consider an IoT network, where the edge node (e.g., a small-cell base station) is deployed at the edge of wireless networks, which is connected to the cloud through fronthaul. There is a total of $F$ randomly distributed IoT sensors. Endowed with caching and computing units, the edge node is able to serve as a relay between data producers (e.g., IoT sensors) and data consumers (e.g., mobile users). More precisely, the edge node maintains a cache unit to aggregate sensing data produced by IoT sensors within its coverage; meanwhile, users can submit their requests to the edge node and retrieve desired content items for data analysis. We consider the scenario, where no direct links present between IoT sensors and users, similar to \cite{niyato2016novel}. Let $\mc F = \{1, 2, \cdots, F\}$ be the set of all indices of sensors. 
In addition, system operation time is assumed to be slotted into a sequence of epochs, i.e., $t = 1, 2, \cdots$. The sensing data cached at the edge node may be dynamically updated as time passed; but every content item, generated by a certain sensor, owns a specific content item index and generation epoch. As previously stated, we term IoT sensing data as {\it transient content}. For instance, content item\footnote{We slightly abuse the notation, and let $f$ denote the index of either an IoT sensor or the associated content item.} $f$ with generation epoch $v_f^t$ means that, at epoch $t$, the caching content item available at the edge node is generated by sensor $f$ at epoch $v_f^t$; the associated generation epoch $v_f^t$ could be reset, once sensor $f$ is determined to upload a new measurement of content item $f$ into the edge node.
\begin{figure}[t]
  \centering
  \includegraphics[scale=1]{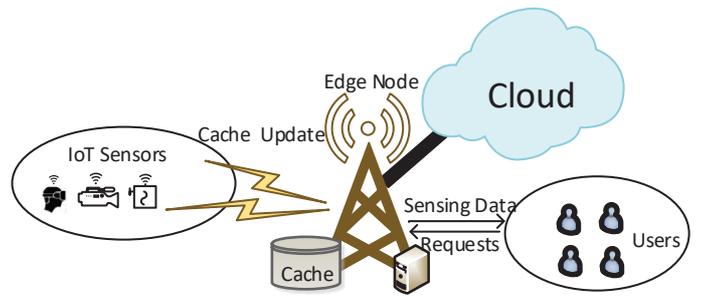}
  \caption{Illustration of an IoT network. }
  \label{system_iot}
\end{figure} 
\subsection{Age of Information}
We introduce a QoS metric, i.e., AoI, to capture how fresh a transient content item is. Particularly, the AoI of content item $f$ at epoch $t$, i.e., $o_f^t$, is defined as how many epochs have elapsed since this content item was generated. Accordingly, it gives rise to the following equation:
\begin{align}
	o_f^t = \max \{t - v_f^t, 1\}, \forall f\in \mc F,
\end{align}
which can take values from $\{1, 2, \cdots, T_{\max}\}$; and $T_{\max}$ denotes the upper limit \cite{abd2019reinforcement}. Let $N_f^t$ be the number of user requests for content item $f$ that are observed by the edge node at epoch $t$. Thus, the average AoI for satisfying user requests at epoch $t$ can be given as follows:
\begin{align}
	O^t = \frac{\sum_{f \in \mc F} o_f^t N_f^t} {\sum_{f \in \mc F} N_f^t}. \label{eq:avgaoi}
\end{align}
As caching content items gradually become outdated, a reasonable cache update is needed to serve user requests at the coming epochs. For ease of discussion, we assume that each IoT sensor is able to upload its sensing data to the edge node within a single epoch. That is, the AoI of a stale content turns to be 1 as long as the associated sensor is chosen to upload the current measurement of this content item. For instance, as depicted in Fig. \ref{fig:aoidemo}, when a transient content is selected to update at epoch $t_1$, the corresponding AoI reduces to 1 at epoch $t_1+1$; otherwise, the AoI will increment by 1 after every epoch.   
\subsection{Transmission Energy Consumption}
To carry out cache updating, IoT sensors need access the edge node via wireless links; and orthogonal frequency bandwidths are allocated to different IoT sensors to avoid interference. Thus, the received SNR at the edge node concerning the $f$-th sensor can be given by:
\begin{align}
	\gamma_f = \frac{P_{f} \chi_f^2 |\kappa_f|^2}{N_0B}, \forall f \in \mc F, 
\end{align}
where $P_f$ denotes transmission power at the $f$-th sensor; coefficient $\chi_f$ denotes the impacts of the path loss and antenna gain; $\kappa_f$ denotes the small-scale fading component; $N_0$ is the power spectrum density; and $B$ denotes the channel bandwidth. We further assume that $|\kappa_f|$ follows the following Rayleigh distribution, i.e., $x\exp(-x^2/2)$. We consider transmissions between IoT sensors and the edge node are effective under the condition that the received SNR surpasses a pre-specified threshold $\gamma_{th}$. Given the storage size of sensing data from the $f$-th sensor as $s_f$ bits, we present the average transmission energy consumption in the following Proposition. 
\begin{Proposition}
	The average transmission energy consumption $\bar E_f$ for the $f$-th IoT sensor ($\forall f \in \mc F)$ to upload its sensing data is given by: 
	\begin{align}
		 \bar E_f = \frac{\log (2) \times P_f s_f}{\log(2) \times r_{th} \exp\left(-\frac{\gamma_{th}}{2\beta_f}\right)  + B\exp\left(\frac{1}{2\beta_f}\right)\mc \rho_f(\gamma_{th} +1)}, \label{eq:avgegy}
	\end{align}
	where function $\rho_f(\cdot)$ is defined as:
	\begin{align}
	 	\rho_f(x) \triangleq \int_{x}^{+\infty} \frac{1}{x} \exp\bigg(-\frac{x}{2 \beta_f}\bigg) dx,
	 \end{align} 
	 and $\beta_f = P_f \chi_f^2 / (N_0B)$; $r_{th}$ denotes the data rate threshold, given as follows: 
	 \begin{align}
	 	r_{th} \triangleq \log_2(1 + \gamma_{th}).
	 \end{align}
\end{Proposition}
\begin{proof}
	Due to the page limit, we only sketch the basic idea here. Since the edge node would fail to decode information if the received SNR is lower than the required threshold $\gamma_{th}$, one can first compute the corresponding outage probability due to channel fading. The next step is to estimate the average transmission delay by calculating the expected transmission data rates. Finally, the average transmission energy consumption is given by the product of transmission power and average transmission delay. 
\end{proof}
\begin{figure}[t]
  \centering
  \includegraphics[scale=1.9]{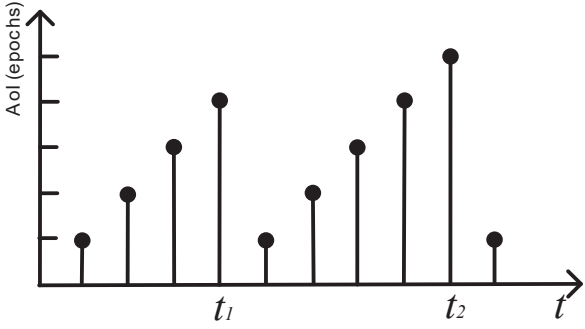}
  \caption{Illustration of the AoI of a specific transient content. Epochs of cache updates are denoted by $t_1$ and $t_2$.}
  \label{fig:aoidemo}
\end{figure}

Clearly, owing to limited battery levels at sensors and massive connections in IoT networks, it is crucial to find the best choice at each epoch to perform cache update so as to obtain an effective tradeoff between AoI and energy cost. Note that, the average AoI (e.g., see \eqref{eq:avgaoi}) highly depends on content popularity distribution that usually exhibits temporal dynamics; meanwhile, energy consumption (e.g., see \eqref{eq:avgegy}) comes to storages of sensing data and channel statistics which are usually inhomogeneous among different sensors. Aware of this, we formulate an MDP based online cache update problem in the following section. 
\section{MDP Problem Formulation}
Typically, an MDP can be described by a tuple $(\mc S, \mc A, \mc P, R, \gamma)$, where $\mc S$ denotes the state space containing all possible states $\bm S$; $\mc A$ denotes the action space collecting all possible actions $\bm A$; $\mc P$ collects transition probabilities $Pr\{\bm S'|\bm S, \bm A\}$; $R$ denotes a reward fed back to the agent after executing an action; and $\gamma \in [0,1)$ is a discount factor. In the online cache update problem, the edge node is anticipated to act as the agent; and we customize the above-mentioned elements in an MDP as follows. 
\begin{itemize}
	\item {\bf State:} At each epoch, the edge node has exact knowledge of the AoI of caching content items as well as user requests. Hence, we define the system state as follows:
	\begin{align}
		\bm S^t = \left(\{o_f^t\}_{f \in \mc F}, \{N_{f}^t\}_{f\in \mc F}\right). 
	\end{align}
	\item {\bf Action:} Similar to study \cite{hatami2020age}, let $\{0, 1, \cdots, F\}$ be the action space\footnote{The proposed framework can be extended to the case where multiple content items are selected at each epoch. This comes at the cost of energy consumption and bandwidth occupation.}. When action $\bm A^t = 0$, it implies that no content is selected; otherwise, we push the new version of content item $\bm A^t \in \mc F$ into the cache unit. As such, it gives rise to the following relation:
	% \begin{align}
	% 	o_f^{t+1} = (o_f^t + 1) \times \mc I(f \neq 
	% 	\bm A^t) + \mc I (f = \bm A^t), \forall f \in \mc F, 
	% \end{align}
    \begin{align}
		o_f^{t+1} = (o_f^t + 1) \times \mc I(f, \bm A^t) + \mc I (f, \bm A^t), \forall f \in \mc F, 
	\end{align}
	where $\mc I(\cdot)$ is an indicator function\footnote{When variables $x,y$ are equal, $\mc I(x,y) = 1$; otherwise, $\mc I(x,y) = 0$}. 
	Clearly, cache update should be carried out after the occurrence of state $\bm S^t$, i.e., after user requests are revealed. Subsequently, the system is expected to transfer to a new state $\bm S^{t+1}$ with transition probability $Pr\{\bm S^{t+1}|\bm S^{t}, \bm A^t\}$.  
	\item {\bf Reward:} The agent in an MDP is required to receive a reward signal $R^{t+1}$ along with the appearance of state $\bm S^{t+1}$. Recall that we aim to minimize the average AoI for satisfying user requests during the upcoming epoch while reducing transmission energy consumption. Hence, it is envisioned to minimize the following average weighted cost, i.e.:
	\begin{align}
		C^{t+1} = \frac{\sum_{f \in \mc F} o_f^{t+1} N_f^{t+1}} {\sum_{f \in \mc F} N_f^{t+1}} + \eta \bar E_f|_{f = \bm A^t}, \label{eq:cost}
	\end{align}
	where the first term on the right-hand side denotes the average AoI at epoch $t+1$; and $\eta \geq 0$ is a constant to strike the balance between the average AoI and energy cost. We denote $\bar E_0 = 0$ for notational convenience.  In accordance with the objective in \eqref{eq:cost}, the reward signal of the proposed framework is designed as follows:
	\begin{align}
		R^{t+1} = - C^{t+1}, \label{eq:insreward}
	\end{align}
	which is supposed to be received at epoch $t+1$. 
	% which coincides with the objective concerned. 
\end{itemize}
To this end, we need to find an optimal policy $\pi^*$ that maximizes the expected discounted cumulative reward, i.e.:
\begin{align}
	\pi^* = \arg\max_{\pi} \mathbb{E} [V^t|\pi], \label{eq:rl} 
\end{align}
where the expectation is over all possibilities of $\{\bm S^t, \bm A^t, \bm S^{t+1}, R^{t+1}\}$, and the discounted cumulative reward is given by:
\begin{align}
	V^t = \sum_{\tau = 0}^{\infty} (\gamma)^{\tau} R^{t+\tau+1}.
\end{align}
Since transition probability, i.e., $Pr\{\bm S'|\bm S, \bm A\}$, is generally difficult to acquire in practical applications, we move on to propose a DRL-based algorithm to address this problem through trial-and-error interactions with the environment.

\section{Proposed DQN Based Algorithm}
In this section, we propose an intelligent cache update design for an IoT network by adopting the state-of-art RL approaches, e.g., DQN. The key idea of this approach is built on utilizing deep neural networks (DNNs) to learn the Q-value function:
\begin{align}
	Q(\bm S, \bm A) = \mathbb[V^t|\bm S = \bm S^t, \bm A = \bm A^t, \pi^*], 
\end{align}
which indicates the expected cumulative reward after executing action $\bm A^t$ and then following policy $\pi^*$; as such, an optimal action can be given by the one attaining the maximum Q-value under current state $\bm S^t$, i.e.:
\begin{align}
	\bm A^* = \arg\max_{\bm A} Q(\bm S^t, \bm A).
\end{align}
According to the {\it Bellman Optimality Equality} in \cite{sutton2018reinforcement}, we have the following recursive result:
\begin{align}
	Q(\bm S^t, \bm A^t) = R^{t+1} + \gamma \max_{\bm A' \in \mc A} Q(\bm S^{t+1}, \bm A').
\end{align}

In general, DQN entails maintaining two networks, i.e., Q-network and target Q-network \cite{mnih2015human}. Specifically, Q-network $Q_{\bth}(\bm S, \bm A)$ can be constructed by a plain DNN, which is parametrized by $\bth$. Readers are referred to \cite{mnih2015human} for greater detail. Target Q-network $Q_{\bth^-}(\bm S, \bm A)$  is parametrized by $\bth^-$, which is designed in the same manner as Q-network and can be used to calculate target values for network training. 

The training procedure is introduced as follows. To aggregate training data, we leverage {\it Replay Buffer (RB)} to store historical experiences, e.g., $\xi^t = (\bm S^t, \bm A^t, \bm S^{t+1}, R^{t+1})$. We assume that {\it RB} can store a total of $N$ experiences; and the stalest experience should be replaced by the fresh one as long as {\it RB} is fully filled. Subsequently, at each iteration, we randomly draw a mini-batch of experiences $\Xi_N$ from {\it RB} to update $\bth$ by minimizing the following loss:
\begin{align}
	\hfill Loss = \mathbb{E}_{\xi^t \sim \Xi_N} \bigg[\big(y^t  - Q_{\bth}(\bm S^t, \bm A^t)\big)^2 \bigg], \label{eq:iotloss}
\end{align} 
where $y^t$ denotes the target value, i.e.:
\begin{align}
	R^{t+1} + \gamma \max_{\bm A'} Q_{\bth^-}(\bm S^{t+1}, \bm A'). \label{eq:target}
\end{align}
Thus, adopting stochastic gradient descent approaches, parameter $\bth$ can be updated as follows\footnote{We consider that parameter $\bth$ and state $\bm S$ are vectorized with proper dimensions.}:
\begin{align}
	\bth \leftarrow \bth - \alpha \bigg[\big(Q_{\bth}(\bm S^t, \bm A^t) - y^t  \big) \nabla_{\bth} Q_{\bth}(\bm S^t, \bm A^t) \bigg], \label{eq:grad}
\end{align}
where $\alpha$ is the learning rate. 
With regard to $\bth^-$, it can be updated by $\bth^- \leftarrow \bth$ every $T_0$ iterations. 
Concerning exploration, we adopt the $\varepsilon$-greedy policy to generate actions: at each epoch, taking an action randomly drawn from $\{0, \cdots, F\}$ with probability $\varepsilon$ whereas taking an optimized action $\bm A^* = \arg\max_{\bm A} Q_{\bth}(\bm S, \bm A)$ with probability $1-\varepsilon$. This is because the success of RL replies to visiting different state-action pairs so as to aggregate sufficient experiences to infer better actions and avoid suboptimal estimations of Q-value function. Finally, we summarize the proposed DRL-based cache update in Algorithm \ref{alg:iotDRL}. 
\begin{algorithm}[!t]
  \caption{DRL-Based Cache Update for IoT Networks}\label{alg:iotDRL}
  \begin{algorithmic}
    \State Initialize $F$, $ \bm S^0 $, $ \varepsilon $, $ \gamma $, $ \alpha $, $ T_0 $
    % \State Initialize experience memory $ EM $
    \State Initialize parameter of Q-network $  \bm{\theta } $, 
    \State Initialize parameter of target Q-network $ \bm{\theta^- }$ 
    \State Initialize {\it RB}
    % \State Initialize the parameter of target ADQN  $ \bm{\theta^- }=\bm{\theta } $ 
    \For{$ t=0,1,2, \cdots $}
    %\State \textit{// Execution}
    \State Input $\bm S^t$ to Q-network and obtain $Q_{\bth}(\bm S^t, \bm A), \forall \bm A \in \mc A$
    \State Execute action $ \bm A^t $ according to $ \varepsilon $-greedy policy
    \State Observe $ \bm S^{t+1}$, $ R^{t+1} $
    % \State Compute $ s_{t+1} $ from $ s_t $, $ \mb a^t $ and $ z_t $
    \State  Store $ \left( \bm S^t, \bm A^t, \bm S^{t+1}, R^{t+1}\right) $ into {\it RB}.
    % \State \textbf{if} Remainder $(t, T_0) == 0$  
    % \textbf{then}  $ I=1 $ \textbf{else} $ I=0 $
    \If {Remainder $(t, T_0) == 0$}
    \State $ I=1 $
    \Else 
    \State $ I=0 $  
    \EndIf
    \State \Call{TrainDQN}{$ \gamma $, $ \alpha $, $ I $, {\it RB}, $ \bm{\theta} $, $ \bm{\theta^-} $}
    \EndFor \\
    \Procedure{TrainDQN}{$ \gamma $, $ \alpha $, $ I $, {\it  RB}, $ \bm{\theta} $, $ \bm{\theta^-} $}
    \State Randomly sample a batch experiences $\Xi_N$ from {\it RB}
    \For{each $\xi^t = \left( {\bm S^t, \bm A^t, \bm S^{t+1}, R^{t+1}} \right) \in \Xi_N $}
    \State Calculate $y^t$ by \eqref{eq:target}
    \EndFor
    \State Calculate $Loss$ by \eqref{eq:iotloss}
    \State $\bth \leftarrow \bth - \alpha \nabla_{\bth} Loss$
    \If{$ I==1 $}
    \State Update the target Q-network by $ \bm{\theta^-} \leftarrow \bm{\theta } $
    \EndIf
    \EndProcedure
  \end{algorithmic}
\end{algorithm}
% We instead use an architecture in which there is a separate output unit for each possible action, and only the state representation is an input to the neural network. The outputs correspond to the predicted Q-values of the individual actions for the input state.
% Since we can only visit one action, we only know the loss for that action (ie. a single output). But as far as I'm aware, we need to have values for all of the outputs in order to train the network. What black magic can you use to get the other output values?
% It seems like a bad idea to get the network to predict the other action-values and feed them back, as it would affect the optimizer. And if you tried to ignore the other outputs and train it as if there were only the one you were currently focused on, you would still affecting the others as they would share edges
% The output of your network should be a Q value for every action in your action space (or at least available at the current state). Then you can use softmax or epsilon-greedy (or other strategies) to select the final action. The network will learn to predict which action should return the maximum reward from your current state. Also we update the network after we have collected a specific amount of experience and we use batches from that experience buffer in order to update the network. We do not feed back any value in the vanilla DQN (no recurrence).
\section{Performance Evaluation}
\subsection{Simulation Setup}
In this section, we investigate the performance of the proposed algorithm. We consider the following IoT network: one edge node is deployed which covers a circle of radius 100 m; a total of 20 IoT sensors are randomly distributed within the coverage of the edge node; the storage size of each content item is randomly drawn from $[50, 100]$ MB. Regarding communications between the edge node and sensors, the channel bandwidth is set as 10 MHz; large-scale fading is specified by the model used in \cite{Wu2019ICC}; noise spectrum density is -172 dBm; and transmission power of IoT sensors is 20 dBm. We consider there are at most 100 users requesting content by following Zipf distributions \cite{Wu2019ICC}, i.e.:
\begin{align}
	p_{f} = \zeta_{f}^{-\kappa}/\sum_{f' \in \mc F} \zeta_{f'}^{-\kappa},
\end{align}
where rank orders of content items, e.g., $\{\zeta_f\}$, are randomly evolving according to a certain probability transmission matrix; and the skewness factor, e.g., $\kappa$, is randomly drawn from $\{0.5, 1, 1.5, 2\}$. In the default setup, we consider that the average AoI and energy consumption are of equal importance, i.e., $\eta =1$. 

To implement algorithm, the Q-network compromises three hidden layers with 512, 256, 128 neurons, respectively. The learning rate is set as 0.001. In $\varepsilon$-greedy policy, we take a random action with probability $\varepsilon = 0.9$; then, $\varepsilon$ is decayed by 0.995 every iteration until 0.05. A mini-batch of 100 experiences is randomly sampled from 
the {\it RB} that has at most 5000 experiences. We update target Q-network every 100 epochs. The discount factor is 0.99. In addition, the following baselines are considered for algorithm comparisons:
\begin{itemize}
	\item {\bf Most Popular Update (MPU):} At every epoch, we consider to update the content item that receives the highest attention from users, i.e.: 
	\begin{align}
		\bm A^t = \arg\max_{f\in \mc F} \frac{N_f^t} {\sum_{f'\in \mc F} N_{f'}^t}.
	\end{align}
	\item {\bf Oracle Update (OU):} We assume that the exact knowledge of up-coming user requests, e.g., $\{N_f^{t+1}\}$, is available at current epoch $t$. Then, we conduct cache update similar to the proposed DQN. This scheme simply serves as a baseline here and is not possible in reality.
	\item {\bf Random Update (RU):}  At every epoch, we randomly select an action from the action space, i.e., $\{0, 1, \cdots, F\}$. This serves as a lower bound for the proposed algorithm.
\end{itemize}  
\begin{figure}[!h]
  \centering
  \includegraphics[scale=0.55]{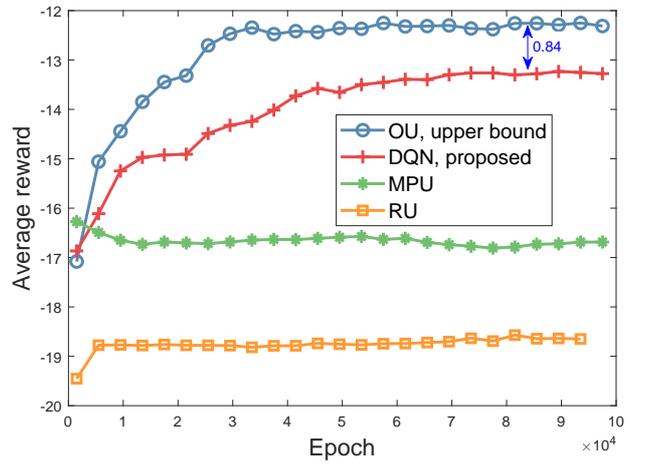}
  \caption{Convergence behavior of the proposed algorithm.}
  \label{fig:iot1}
\end{figure}
\subsection{Learning Curves} 
We first illustrate the learning curves of the proposed algorithm in Fig. \ref{fig:iot1}. For illustrative purposes, we show the moving average results, which are attained by averaging rewards over $N = 10000$ epochs, i.e., $\sum_{\tau = t-N+1}^t R^{\tau}/N$.
% We adopt the moving average reward as the performance metric,
Moreover, the resulting reward is less than zero since it is defined as the negative value of average weighted cost; see \eqref{eq:insreward}.
It can be observed that, in the initial stage, the learning curve of the proposed algorithm increases rapidly and surpasses the curve of MPU; after 10000 epochs, it continuously increases; as more states are visited with time elapses, the resulting reward gradually grows larger, eventually converging to a higher level than that of MPU and RU. 
Note that, the average reward achieved by DQN is only 0.84 less than the upper bound, i.e., achieved by OU, whereas it exceeds the results of MPU and RU by 28.8\% and 40.4\%, respectively. This result implies that the proposed DQN algorithm is able to not only track the dynamics of user preferences towards content items but also adapt to energy costs among different sensors in physical-layer transmissions. These observations corroborate the remarkable performance of the proposed algorithm. 

\subsection{Tradeoff Between AoI and Energy Cost}
To investigate the tradeoff between AoI and energy cost, we vary factor $\eta$ and plot the results of average total cost, the average AoI, and average transmission energy in Figs. \ref{fig:iot2} - \ref{fig:iot3}, respectively. As can be seen in Fig. \ref{fig:iot2}, under different $\eta$, the proposed algorithm performs very close to the upper bound, and significantly outperforms other baselines, i.e., MPU and RU. For instance, when $\eta = 20$, the achieved average weighted cost can be reduced by 54.9\% and 44.3\% in comparison to MPU and RU, respectively. These findings again demonstrate the advantages of the proposed DRL-based cache update  and also validate its robustness toward factor $\eta$.   
\begin{figure}[t]
  \centering
  \includegraphics[scale=0.55]{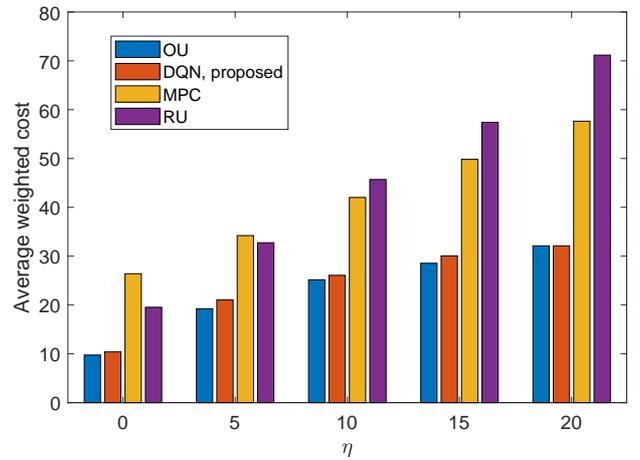}
  \caption{Average weighted cost.}
  \label{fig:iot2}
\end{figure}

\begin{figure}[h]
  \centering
  \includegraphics[scale=0.55]{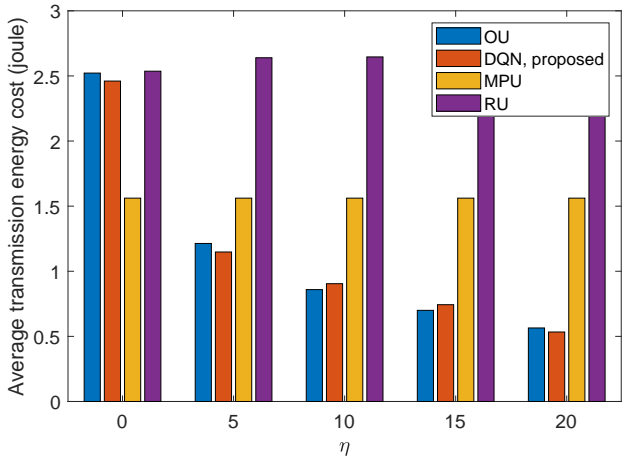}
  \caption{Average transmission energy cost.}
  \label{fig:iot4}
\end{figure}
\begin{figure}[h]
  \centering
  \includegraphics[scale=0.55]{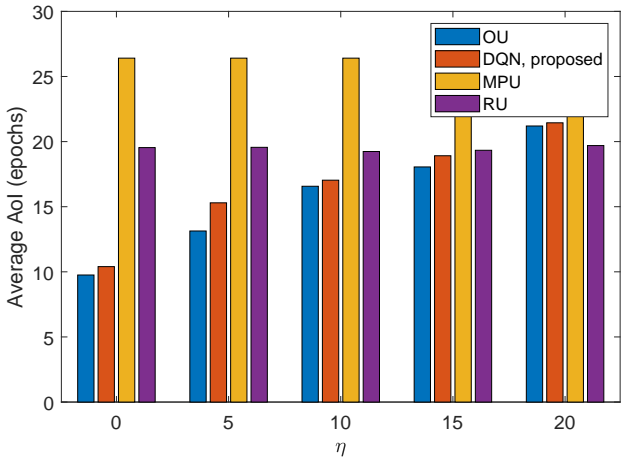}
  \caption{Average AoI.}
  \label{fig:iot3}
\end{figure}
Clearly, a larger $\eta$ implies that we pay more attention to minimizing energy consumption. As can be seen in Fig. \ref{fig:iot4}, average energy costs, achieved by DQN and OU, drop off as $\eta$ becomes larger while other baselines only witness rather slight fluctuations without aware of AoI and energy cost tradeoff. On the other hand, in Fig. \ref{fig:iot3}, we observe a reverse trend; that is, the caching content items become staler as $\eta$ grows larger. We conjecture that cache update could happen less frequently such that sensing data available at the cache unit becomes outdated. It is worth pointing out that when $\eta = 5$, the achieved average AoI (by DQN) only degrades 2.41 while we can see a notable reduction of 52.7\% in energy cost in contrast with the case $\eta = 0$. If we continue increasing $\eta$, the reduction in energy consumption is quite limited yet at the cost of worsening AoI. Hence, it is necessary to choose a suitable $\eta$ to balance AoI and energy cost in real applications.

\section{Conclusion}
In this paper, we have put forth a deep reinforcement learning framework for online cache updating in IoT networks under dynamic content popularity. 
The objective of this framework is to minimize the weighted average AoI plus energy cost. We have developed a characterization of transmission energy consumption at IoT sensors. Through trial-and-error explorations, the proposed DQN algorithm is capable of adapting to temporal dynamics of user requests as well as the inhomogeneous content size and channel statistics among different IoT sensors. Simulation results have been presented to demonstrate the effectiveness of the proposed design and reveal how transmission energy consideration compromises data freshness.
\bibliographystyle{IEEEtran}
\bibliography{references}

\begin{thebibliography}{10}
\providecommand{\url}[1]{#1}
\csname url@samestyle\endcsname
\providecommand{\newblock}{\relax}
\providecommand{\bibinfo}[2]{#2}
\providecommand{\BIBentrySTDinterwordspacing}{\spaceskip=0pt\relax}
\providecommand{\BIBentryALTinterwordstretchfactor}{4}
\providecommand{\BIBentryALTinterwordspacing}{\spaceskip=\fontdimen2\font plus
\BIBentryALTinterwordstretchfactor\fontdimen3\font minus
  \fontdimen4\font\relax}
\providecommand{\BIBforeignlanguage}[2]{{%
\expandafter\ifx\csname l@#1\endcsname\relax
\typeout{** WARNING: IEEEtran.bst: No hyphenation pattern has been}%
\typeout{** loaded for the language `#1'. Using the pattern for}%
\typeout{** the default language instead.}%
\else
\language=\csname l@#1\endcsname
\fi
#2}}
\providecommand{\BIBdecl}{\relax}
\BIBdecl

\bibitem{madakam2015internet}
S.~Madakam, V.~Lake, V.~Lake, V.~Lake \emph{et~al.}, ``Internet of things
  ({IoT}): A literature review,'' \emph{Int. J. Comput. Commun.}, vol.~3,
  no.~05, p. 164, May 2015.

\bibitem{chen2019artificial}
M.~Chen, U.~Challita, W.~Saad, C.~Yin, and M.~Debbah, ``Artificial neural
  networks-based machine learning for wireless networks: A tutorial,''
  \emph{IEEE Commun. Surveys Tuts.}, vol.~21, no.~4, pp. 3039--3071,
  Fourthquarter 2019.

\bibitem{wu2020joint}
X.~Wu, Q.~Li, X.~Li, V.~C. Leung, and P.~Ching, ``Joint long-term cache
  updating and short-term content delivery in cloud-based small cell
  networks,'' \emph{IEEE Trans. Commun.}, vol.~68, no.~5, pp. 3173 -- 3186, May
  2020.

\bibitem{wu2019jointMDS}
X.~Wu, Q.~Li, V.~C. Leung, and P.~Ching, ``Joint fronthaul multicast and
  cooperative beamforming for cache-enabled cloud-based small cell networks: An
  {MDS} codes-aided approach,'' \emph{IEEE Trans. Wireless Commun.}, vol.~18,
  no.~10, pp. 4970--4982, Oct. 2019.

\bibitem{li2018hierarchical}
X.~Li, X.~Wang, P.-J. Wan, Z.~Han, and V.~C. Leung, ``Hierarchical edge caching
  in device-to-device aided mobile networks: Modeling, optimization, and
  design,'' \emph{IEEE J. Sel. Areas Commun.}, vol.~36, no.~8, pp. 1768--1785,
  June 2018.

\bibitem{kaul2012real}
S.~Kaul, R.~Yates, and M.~Gruteser, ``Real-time status: How often should one
  update?'' in \emph{Proc. IEEE INFOCOM}, Mar. 2012, pp. 2731--2735.

\bibitem{sun2017update}
Y.~Sun, E.~Uysal-Biyikoglu, R.~D. Yates, C.~E. Koksal, and N.~B. Shroff,
  ``Update or wait: How to keep your data fresh,'' \emph{IEEE Trans. Inf.
  Theory}, vol.~63, no.~11, pp. 7492--7508, Nov. 2017.

\bibitem{wang2020federated}
X.~Wang, C.~Wang, X.~Li, V.~C. Leung, and T.~Taleb, ``Federated deep
  reinforcement learning for internet of things with decentralized cooperative
  edge caching,'' \emph{IEEE Internet Things J.}, to appear, 2020.

\bibitem{ma2020deep}
M.~Ma and V.~W. Wong, ``A deep reinforcement learning approach for dynamic
  contents caching in {HetNets},'' \emph{arXiv preprint arXiv:2004.07911},
  2020.

\bibitem{ceran2018reinforcement}
E.~T. Ceran, D.~G{\"u}nd{\"u}z, and A.~Gy{\"o}rgy, ``A reinforcement learning
  approach to age of information in multi-user networks,'' in \emph{Proc. IEEE
  PIMRC}, Sept. 2018, pp. 1967--1971.

\bibitem{hatami2020age}
M.~Hatami, M.~Jahandideh, M.~Leinonen, and M.~Codreanu, ``Age-aware status
  update control for energy harvesting {IoT} sensors via reinforcement
  learning,'' \emph{arXiv preprint arXiv:2004.12684}, 2020.

\bibitem{ceran2019reinforcement}
E.~T. Ceran, D.~G{\"u}nd{\"u}z, and A.~Gy{\"o}rgy, ``Reinforcement learning to
  minimize age of information with an energy harvesting sensor with {HARQ} and
  sensing cost,'' in \emph{Proc. IEEE INFOCOM Wkshps}, Apr. 2019, pp. 656--661.

\bibitem{niyato2016novel}
D.~Niyato, D.~I. Kim, P.~Wang, and L.~Song, ``A novel caching mechanism for
  {Internet} of things ({IoT}) sensing service with energy harvesting,'' in
  \emph{Proc. IEEE ICC}, May 2016, pp. 1--6.

\bibitem{abd2019reinforcement}
M.~A. {Abd-Elmagid}, H.~S. {Dhillon}, and N.~{Pappas}, ``A reinforcement
  learning framework for optimizing age-of-information in {RF}-powered
  communication systems,'' \emph{IEEE Trans. Commun.}, vol.~68, no.~8, pp.
  4747--4760, Aug., 2020.

\bibitem{sutton2018reinforcement}
R.~S. Sutton and A.~G. Barto, \emph{Reinforcement Learning: An
  Introduction}.\hskip 1em plus 0.5em minus 0.4em\relax MIT Press, 2018.

\bibitem{mnih2015human}
V.~Mnih, K.~Kavukcuoglu, D.~Silver, A.~A. Rusu, J.~Veness, M.~G. Bellemare,
  A.~Graves, M.~Riedmiller, A.~K. Fidjeland, G.~Ostrovski \emph{et~al.},
  ``Human-level control through deep reinforcement learning,'' \emph{Nature},
  vol. 518, no. 7540, p. 529, Feb. 2015.

\bibitem{Wu2019ICC}
X.~Wu, Q.~Li, X.~Li, V.~C. Leung, and P.~Ching, ``Joint long-term cache
  allocation and short-term content delivery in green cloud small cell
  networks,'' in \emph{Proc. IEEE ICC}, May 2019.

\end{thebibliography}
\end{document}